\documentclass{llncs}
\usepackage{amsmath}
\usepackage{amsfonts}

\newcommand{\cc}{\mbox{$\mathcal C$}}
\newcommand{\nn}{\mathbb{N}}

\newtheorem{fact}{Fact}

\titlerunning{On the Combinatorial Power of the Weisfeiler-Lehman Algorithm}
\title{On the Combinatorial Power of the Weisfeiler-Lehman Algorithm}
\author{Martin F\"urer%
\thanks{This work was partially supported by NSF Grant CCF-1320814.
Part of this work has been done while visiting Theoretical Computer Science, ETH Z\"urich, Switzerland}}

 \institute{Department of Computer Science and Engineering\\
Pennsylvania State University \\ 
 University Park, PA, USA\\
  \email{furer@cse.psu.edu}}
  \authorrunning{Martin F\"urer}

\begin{document}

\maketitle \thispagestyle{plain} \pagestyle{plain}

\begin{abstract}
The classical Weisfeiler-Lehman method WL[2] uses edge colors to produce a powerful graph invariant. It is at least as powerful in its ability to distinguish non-isomorphic graphs as the most prominent algebraic graph invariants. It determines not only the spectrum of a graph, and the  angles between standard basis vectors and the eigenspaces, but even the angles between projections of standard basis vectors into the eigenspaces. Here, we investigate the combinatorial power of WL[2]. For sufficiently large $k$, WL$[k]$ determines all combinatorial properties of a graph. Many traditionally used combinatorial invariants are determined by WL$[k]$ for small $k$. We focus on two fundamental invariants, the number of cycles $C_p$ of length $p$, and the number of cliques $K_p$ of size $p$. We show that WL[2] determines the number of cycles of lengths up to 6, but not those of length 8. Also, WL[2] does not determine the number of 4-cliques.
 \end{abstract}


\section{Introduction}
\subsection{Weisfeiler-Lehman Method}
Two graphs are isomorphic, if there is a bijection of their vertices mapping edges to edges and non-edges to non-edges. An automorphism of a graph is an isomorphism from the graph to itself.  The graph isomorphism problem is closely connected to the graph automorphism problem. Two connected graphs $G$ and $G'$ with disjoint vertex sets are isomorphic, iff their union has an automorphism mapping one vertex of $G$ into a vertex of $G'$. Obviously in this case, all vertices of $G$ are mapped to vertices of $G'$.

The most natural and most practical way to detect that two graphs are not isomorphic is vertex classification \cite{ReadC77}. The idea is to give different colors to two vertices whenever it is obvious that neither of them can be mapped to the other one by an isomorphism. Thus vertex classification could start by coloring the vertices by their degree. One can easily go further. If one vertex $u$ has more neighbors of a certain degree than another vertex $v$, then obviously $u$ and $v$ should also be colored differently. 

A simple way to capture these observations, is to start with all vertices of $G=(V,E)$ having the same color, and then refining the coloring in rounds. In round $i+1$, vertices $u$ and $v$ receive different colors, if they already had different colors in round $i$, or if the multisets of colors of neighbors of $u$ and $v$ in round $i$ are different.\footnote{A multiset differs from a set by assigning a positive integer multiplicity to each element.} During each round some color classes are split into two or more classes, until this process stops after at most $n=|V|$ rounds.
Nowadays, this method of vertex classification is also known as WL[1].
It is at the heart of all software tools for graph isomorphism testing.

The classical Weisfeiler-Lehman method WL[2] \cite{WeisfeilerL76}, classifies edges in a similar way. Still, it is a bit more involved. 
In fact, all ordered pairs of vertices are classified, not just the edges. In other words, we can think of handling a complete directed graph with colored edges, including self-loops in all vertices.

At the start, the edges of the complete graph are partitioned into 3 color classes: the previous edges, the previous non-edges, and the self-loops. In round $i+1$ every directed edge $(u,v)$ is colored with a pair whose first component is its previous color, and whose second component is the multiset of all pairs of previous colors on paths of length 2 from $u$ to $v$. 
In each round, the actually occurring colors are lexicographically ordered and replaced by an initial segment of the natural numbers.
This time, after $O(n^2)$ rounds, the algorithm stops, because a {\em stable coloring} is reached, i.e., no color class of edges is further divided.

Sometimes, it is useful to keep for each round the mapping assigning to each detailed color (pair of old color and some multiset) a simplified color (small integer). We refer to this information as the definition of colors.

It has been noticed that WL[2], has a natural $k$-dimensional extension WL$[k]$ by various researchers, including some authors of \cite{CaiFI92} who tried to prove that WL$[k]$ solves the graph isomorphism problem for graphs of degree at most $k$. It seems that the first published definition of WL$[k]$ has been in \cite{BabaiM80}. 
The CFI algorithm \cite{CaiFI92} has introduced and popularized the term WL$[k]$ at the suggestion of Babai as an editor to honor the influence of Weisfeiler and Lehman \cite{WeisfeilerL76} towards the development of this algorithm.  

Weisfeiler and Lehman did not use the WL$[k]$ algorithm, but extended WL$[2]$ by individualizing a sequence of vertices.
A sequence $v_1, \dots, v_{\ell}$ is individualized by giving a unique color to each vertex of the sequence before the WL$[k]$ algorithm starts. Note that WL$[k+\ell]$ is at least as powerful as doing WL$[k]$ for every possible individualization of $\ell$ vertices.

WL$[k]$ is defined as follows. The initial color $W^0(v_1, \dots , v_k)$ is according to the isomorphism type of $(v_1, \dots , v_k)$.
To be precise, $(u_1, \dots , u_k)$ is isomorphic to $(v_1, \dots , v_k)$ if
\begin{itemize}
\item 
for all $i,j$, $u_i = u_j$, iff $v_i = v_j$, and
\item
for all $i,j$, $\{u_i,u_j\} \in E$, iff $\{v_i,v_j\} \in E$.
\end{itemize}
For each coloring $f: V^k \rightarrow C$ and each $w \in V$, define the operation
\begin{eqnarray*}
\lefteqn{\mbox{sift}(f, (u_1, u_2, \dots, u_k), w))} \\
& = & \langle f(w, u_2, \dots, u_k), f(u_1, w, \dots, u_k), \dots,
f(u_1, u_2, \dots, w) \rangle.
\end{eqnarray*}
Hence, $\mbox{sift}(W^i, (u_1, u_2, \dots, u_k), w))$ is the $k$-tuple of $W^i$ colors of the $k$-tuples arising from substituting $w$ in turn for each of the $k$ positions in $(u_1, u_2, \dots, u_k)$. Thus, intuitively in each round of WL[3], triangles $T$ are colored by the multiset of the triples of colors used on the triangular faces of the tetrahedra with one face being $T$. To be precise, actually ordered triples of vertices are used instead of triangles.
Now the next color of $(u_1, u_2, \dots, u_k)$ is 
\[(f(u_1, u_2, \dots, u_k), \mbox{multiset}\{
\mbox{sift}(f, (u_1, u_2, \dots, u_k), w) \mid w \in V \}).\]

It should be noticed that for every $k$, WL$[k+1]$ is at least as powerful as WL$[k]$, because every stable coloring \cc\ of the $k+1$-tuples defines a stable coloring $\cc'$ of the $k$-tuples by $\cc'(v_1, \dots , v_k) = \cc(v_1, \dots, v_k, v_k)$, i.e., by just repeating the last component. Thus for example WL[2] does not only color the edges, but also the vertices. The color of a vertex $v$ shows up as the color of the self-loop at $v$.
The stable partition of the $k$-tuples of vertices produced by WL$[k+1]$ is at least as fine as that produced by WL$[k]$, because WL$[k]$ produces the coarsest stable partition of the $k$-tuples.

\subsection{Graph Invariants}
A {\em graph invariant} is any function defined on graphs whose value is constant on classes of isomorphic graphs. In particular, the value does not depend on the enumeration of the vertices. In other words, a graph invariant is a function defined on adjacency matrices whose value does not change, when the same permutation is applied to the rows and columns of an adjacency matrix.

Many simple combinatorial graph invariants are often used to quickly conclude that two graphs are non-isomorphic. Some such invariants are, the number of vertices $n$, the number of edges $m$, the number of triangles, the degree (maximum number of neighbors of any vertex), the multiset of degrees of vertices. Graph invariants can also be just boolean properties like being bipartite, being connected, being acyclic, or containing a given graph as a subgraph or induced subgraph. 

Some more complicated invariants are obtained by counting cliques and cycles. 
We use the words {\em path} and {\em cycle} to refer to a simple path or simple cycle respectively (i.e., an open or closed vertex disjoint path).
More precisely, we refer to the set of their edges. Thus, e.g., a $K_3$ consists of 1 cycle. 

Let $C_k^v$ be the number of $k$-cycles with one vertex being $v$. Now the multiset of all $C_k^v$ for fixed $k$ and varying over all $v \in V$ is a nice invariant. Similar invariants are obtained by varying over all edges instead of vertices, and by considering $k$-cliques instead of $k$-cycles. We will mainly focus on the graph invariants $\#k$-cliques, the total number of subsets of $k$ vertices forming a complete graph, and $\#k$-cycles, the number of cycles of length $k$ which are occurring in the given graph. 

The ultimate combinatorial invariant is obtained by the WL$[k]$ method. Its strength increases with $k$, and it determines the isomorphism type for $k=n$. We call the invariant WL$[k]$ too. The invariant consists of the multiset of colors of $k$-tuples in the stable refinement, together with all the definitions of colors occurring during the coloring rounds. 

A graph invariant identifies a graph $G$ in a class of graphs, if all graphs in the class with the same invariant as $G$ are isomorphic to $G$. In other words, up to isomorphisms, $G$ is the only graph in the class with this invariant. A graph invariant identifies a graph $G$, if it identifies $G$ in the class of all graphs. An invariant identifies a class of graphs, if it identifies all graphs $G$ of this class in the class of all graphs.
For example, the spectrum does not identify the trees, while the lexicographically first adjacency matrix (varying over all enumerations of the vertices) identifies all graphs. Of course, no fast algorithm is known to compute the lexicographically first adjacency matrix of a graph. 

WL$[n]$ trivially identifies all graphs of size at most $n$. On the other hand, even WL[1] is sufficient to identify almost all graphs  \cite{BabaiES80}. In fact, for almost all graphs, WL[1] stops after the second round with all vertices receiving distinct colors. The remaining graphs can be handled sufficiently fast to obtain an $O(n^2)$ expected time algorithm \cite{BabaiK79} (linear in the input size of a random graph). Even almost all regular graphs can be identified by WL[2], resulting in a linear expected time algorithm for identifying the regular graphs \cite{Kucera87}. In general, it is difficult to find instances of graphs that are not easily identified. One source of such graphs are strongly regular graphs, which are the graphs where WL[2] stops immediately after assigning the initial colors without doing any refinements.

Algebraic graph invariants are among the most widely studied invariants.
Examples of algebraic invariants are the spectrum (the multiset of eigenvalues of the adjacency matrix), the Laplacian spectrum, the multiset of angles of the standard basis vectors with the eigenspace for a given eigenvalue. The standard basis vectors are those with a component 1 in one vertex and components 0 in all other vertices. Note that multisets rather than $n$-tuples have to be used here, because in general no ordering of the vertices can be defined in an invariant way.

The standard algebraic graph invariants have a low distinguishing power, compared to strong combinatorial invariants. Already WL[2] determines the spectrum. The WL[2] color of a vertex determines the lengths of the projections of its standard basis vectors into the eigenspaces, and the WL[2] color of an edge determines the angle between the projections of its endpoints \cite{Furer95} (see also \cite{Furer2010}). The spectrum of the $k$-th power of a graph $G$ is more powerful than the spectrum of $G$ itself, but not as powerful as WL[2k] \cite{AlzagaIP2010}.

\subsection{The Graph Isomorphism Problem}

The graph isomorphism problem, i.e., testing whether two graphs are isomorphic is not known to be in P, but not believed to be NP-complete, as this would have strange consequences like the collapse of the polynomial hierarchy. Babai \cite{Babai2015,Babai2016} has recently shown the graph isomorphism problem to be in pseudo-polynomial time (i.e., in time $2^{(\log n)^{O(1)}}$). This result builds on the milestone work of Luks \cite{Luks82}, who proved that graphs of bounded degree can be tested in polynomial time. These results rely heavily on group theoretical methods. Since the early eighties the author was involved in an oral debate, whether combinatorial methods could solve the bounded degree case too. In particular, it was open whether WL$[k]$, with $k$ being the degree of the graph, could solve the bounded degree graph isomorphism problem. This would be a very natural algorithm, running in polynomial time, or more precisely, in time $O(n^{k+1} \log n)$.
It was even not clear whether a constant $k$ would be sufficient for all graphs. Some  support for this possibility was provided by the result that WL[5] always makes at least some progress  \cite{Cameron80,GolfandK78,KlinPR88} except for some known trivial cases.

These questions have been answered by the CFI result \cite{CaiFI92}. It shows that WL$[k]$ requires $k = \Omega(n)$ in order to identify all graphs of size $n$. We now introduce this construction, 
since we use it for our proofs later. It starts with an arbitrary graph $H$ called the global graph. For the $\Omega(n)$ result, $H$ has to be an expander graph, but any low degree graph can be used for the construction. Here we only describe the interesting case of $H$ being regular of degree 3. We show how to produce two similar graphs $G$ and $\widetilde{G}$ from $H$. The graphs $G$ and $\widetilde{G}$ are not isomorphic, but WL[2] uses edge colors with the same multiplicities.
\begin{enumerate}
\item 
Every vertex $v$ of $H$ is replaced by 4 vertices $v_0,v_1,v_2,v_3$ of $G$ arranged counterclockwise in the corners of a square, but without the edges of the square. Note, that there are 3 partitions of $\{v_0,v_1,v_2,v_3\}$ into two subsets of vertices of size two:
\begin{enumerate}
\item Bottom $\{v_0, v_1\}$, Top $\{v_2, v_3\}$,
\item Left $\{v_0, v_3\}$, Right $\{v_1, v_2\}$,
\item Slash $\{v_0, v_2\}$, Backslash $\{v_1, v_3\}$,
\end{enumerate}
\item
Consider every edge $\{u,v\}$ of $H$ to consist of 2 directed edges $(u,v)$ and $(v,u)$. For every vertex $u$ of $H$ label the 3 outgoing edges in an arbitrary way with the 3 partitions a, b, c, from above.
\item
Now introduce 8 edges of $G$ to replace every edge $\{u,v\}$ of $H$. For example, if $(u,v)$ is labeled {\bf a}, and $(v,u)$ is labeled {\bf b}, then the bottom $u$-nodes are connected to the left $v$-nodes, and the top $u$-nodes are connected to the right $v$-nodes. In other words, the edge $\{u_i,v_j\}$ is introduced, if either $i \in \{0,1\}$ and $j \in \{0,3\}$, or  $i \in \{2,3\}$ and $j \in \{1,2\}$.
\end{enumerate}

Finally, $\widetilde{G}$ is constructed from $G$ by picking an arbitrary edge of $H$ and flipping the corresponding connections in $G$. In the previous example, the bottom $u$-nodes would be connected to the right $v$-nodes, and the top $u$-nodes would be connected to the left $v$-nodes.

\begin{fact}
The location of a flip is undefined.
It can easily be moved from an edge incident to a vertex $v$ of $H$ to any of the other edges incident on $v$ by doing a Bottom-Top ($\{v_0,v_1\} \leftrightarrow \{v_2,v_3\}$) exchange and/or a Left-Right ($\{v_0,v_3\} \leftrightarrow \{v_1,v_2\}$) exchange. In several steps, the flip can be moved to any edge in the same connected component.
\end{fact}

\begin{fact}
Only the parity of the flips matter.
 If $G$ is manipulated by introducing an even number of flips, we obtain a graph isomorphic to $G$.  If $G$ is manipulated by introducing an odd number of flips, we obtain a graph isomorphic to $\widetilde{G}$. 
\end{fact}

\subsection{Summary of Main Results}
In this paper, we study the power of WL[2] in comparison with the graph invariants $\#k$-cliques and $\#k$-cycles for different values of $k$.
In the next section we study the positive results. For cliques we only have the trivial result that 3-cliques are identified by WL[2]. For cycles, astonishingly WL[2] is much more powerful. Of course, 3-cycles and 4-cycles are identified, but surprisingly also 5-cycles and 6-cycles are identified.
Section 3 contains the negative result for 4-cliques, and Section 4 is devoted to the negative result for 8-cycles.

\section{Positive Results}

Recall that we use the words {\em path} and {\em cycle} to refer to the set of edges of a simple path or cycle. Walks (not necessarily simple paths) and closed walks are not so interesting in our context. For example, for regular graphs, their numbers are determined by the graph size and the degree.
It is not hard to see that WL[2] can easily count walks and closed walks of any length. More interesting is the task of counting (simple) paths and cycles.

We say that WL$[k]$ counts the number of $j$-cycles or solves the problem $\#j$-cycles, if it produces a multiset of colors (including their definitions) that is only produced for graphs that have exactly the same number of $j$-cycles. In the same way, we define WL$[k]$ counting the number of $j$-cliques or solving the $\#j$-clique problem. Similarly, we say that an edge $\{u,v\}$ knows a certain number, if the color of $(u,v)$ and its definition determines that number.

\begin{theorem}
 WL[2] counts the number of triangles.
\end{theorem}

\begin{proof}
Obviously, WL[2] trivially counts the number of triangles. After 1 round, every edge knows the number of triangles it is involved in. Therefore, the multiset of all colors of edges determines the total number of triangles. 
 If $c_j$ edges are involved in $j$ triangles, then the total number of triangles is $\frac{1}{3} \sum_{j=1}^{n-2} c_j j$.
\qed \end{proof}

For $\#k$-cliques, this trivial positive result is all we get. For $\#k$-cycles we can do much better. But first we look at the problem of counting the number of  paths of length 4 between a given pair of vertices. This could easily be used to show that WL[2] also counts the total number of paths of length 4. We don't treat counting the paths of length 5, as it can be done along the lines of the $\#6$-cycles problem. Counting the paths of length $k<5$ is easy.

We say that a coloring algorithm WL$[k]$ computes a function or decides a property of graphs, if the multiset of stable colors of the $k$-tuples determines the function value or the property respectively. This means that whenever two graphs have the same multiset of colors, then they agree in the function value or the property respectively. 

Here, having the same multiset of colors of $k$-tuples can be defined in two equivalent ways. 
\begin{itemize}
\item 
The two graphs are colored simultaneously, i.e., when the names of the colors are reduced to small integers, a small integer abbreviates the same long name in both graphs.
In this scenario, there is no need to retain the color definitions.
\item 
Each graph is colored separately, but the definitions of the colors are included. The two graphs must have the same number of equally defined colors. 
Here, it is important to include an additional color definition, when the color partition is already stable.
A key example consists of two strongly regular graphs with the same number of vertices and edges, but with different parameters $\lambda$ (number of common neighbors of adjacent vertices) and $\mu$ (number of common neighbors of non-adjacent vertices).
In each graph the edge coloring is stable from the beginning, as even the first refinement round has no effect. But the new colors in the two graphs have different definitions.
\end{itemize}

Similarly, we define what it means for a $k$-tuple {\em to know} a function value or a property. It means that WL$[k]$ colors the $k$-tuple with a color (including its definition) that only shows up when the function has this value or the graph has this property respectively.

\begin{lemma} \label{lem:p4}
 WL[2] can count the number of paths of length 4 between any pair of vertices.
\end{lemma}

\begin{proof}
We show that every edge $(u,v)$ knows the number of paths of length 4 from $u$ to $v$. 
Let $p_{uv}^{\ell}$ be the number of paths of length $\ell$ from $u$ to $v$. Every vertex pair $(u,v)$ knows $p_{uv}^{1}$ from the start and $p_{uv}^{2}$ after 1 round.
For all $\ell_1, \ell'_1 \in \{0,1\}$ and $\ell_2, \ell'_2 \in \nn$, after 2 rounds, $(u,v)$ knows
\[
n_{\ell_1 \ell_2 \ell_1' \ell_2'} :=  \#\{w \mid w \notin \{u,v\} \land p_{uw}^i = \ell_i \land p_{wv}^i = \ell'_i \mbox{ for $i \in \{1,2\}$}\}.
\]
Then $(u,v)$ knows the number of paths of length 4 from $u$ to $v$, which is
\[ \sum_{\ell_1, \ell_2, \ell_1', \ell_2'} n_{\ell_1 \ell_2 \ell_1' \ell_2'} 
	(\ell_2 - p_{uv}^1 \, \ell'_1) (\ell'_2 - p_{uv}^1 \, \ell_1) 
	- \!\!\!\!\!\! \sum_{x \in V \setminus \{u,v\}} \!\! p_{ux}^1 (d(x)-2) p_{xv}^1, \]
where $d(w)$ is the degree of vertex $w$. Of course, $(u,w)$ knows $d(w)$ after 1 round.

For the correctness of this formula, notice when combining all paths of length~2 from $u$ to $w$ with all paths of length 2 from $w$ to $v$, we also encounter two kinds of undesirable walks. First, we don't allow the paths of length 2 from $u$ to $w$ through $v$ and from $w$ to $v$ through $u$. Finally, we subtract all walks $u, x, w, x, v$ for any vertex $x$.
\qed \end{proof}

\begin{theorem}
 WL[2] solves $\#k$-cycles for $k \leq 6$.
\end{theorem}

\begin{proof}
 For $k=4$ the result is easy. Every edge $e$ can count the number of quadrangles of which it is a diagonal. In one round the edge $e=(u,v)$ knows the number of common neighbors $n(e)=p_{uv}^2$. Then the over all number of quadrangles is $\frac{1}{2} \sum _{e \in E} \binom{n(e)}{2}$.

For $k=5$ the result follows from the lemma. Every pair $(u,v)$  knows the number $p_{uv}^4$ of paths of length 4 from $u$ to $v$, and it knows whether $u$ and $v$ are adjacent. Thus $(u,v)$ knows in how many 5-cycles it is involved.

$k=6$ is the interesting case. Any closed path $v_0, v_1, \dots, v_5, v_0$ of length 6 can be broken down into a path of length 4 and a path of length 2 from $v_0$ to $v_4$.

In order to count the number of 6-cycles, we count for every fixed pair $(v_0,v_4)$  with $v_0 \neq v_4$ the number of 6-cycles $H = v_0, v_1, v_2, v_3, v_4, v_5, v_0$ which are subgraphs of  the given graph $G$ for variable $v_1, v_2, v_3, v_5$. From now on the pair $(v_0,v_4)$ is fixed. Let $\#H$ be the number of such subgraphs $H$.
Let $\#H(*)$ be the number of subgraphs consisting of a path $v_0, v_1, v_2, v_3, v_4$ of length 4 and a path $v_0, v_5, v_4$ of length 2 where $v_5$ might possibly be identified with $v_1$, $v_2$, or $v_3$.
Let $H[u_1=w_1, \dots, u_p =w_p]$ be any graph obtained from a graph of type $H$ by identifying $u_i$ with $w_i$ for $1 \leq i \leq p$. Then 
\[ \#H = \#H(*) - \#H(v_1 = v_5) - \#H(v_2 = v_5) - \#H(v_3 = v_5). \]

After 1 round, the pairs $(v_0, v_2)$, $(v_2, v_4)$, and $(v_0, v_4)$ all know the number of paths of length 2 between them.
After 2 rounds, the pair $(v_0, v_4)$, also knows the number $p_{v_0 v_4}^4$ of paths $v_0, v_1, v_2, v_3, v_4$ of length 4 from $v_0$ to $v_4$ by Lemma~\ref{lem:p4}.

Thus after 2 rounds, the pair $(v_0, v_4)$ knows $\#H(*) = p_{v_0 v_4}^2  p_{v_0 v_4}^4$, which is the number of pairs of paths  $v_0, v_1, v_2, v_3, v_4$ and $v_0,v_5,v_4$.
Not every such pair of paths can be combined to a 6-cycle.
We have to subtract the number of cases, where $v_5$ is equal to one of the vertices $v_1$, $v_2$, or $v_3$.

Let's compute $\#H(v_2 = v_5)$. 
After 1 round $(v_0,v_2)$ and $(v_2,v_4)$ know the number of triangles in which they participate. These numbers are $p_{v_0 v_2}^1 p_{v_0 v_2}^2$ and $p_{v_2 v_4}^1 p_{v_2 v_4}^2$ respectively. After 2 rounds, $(v_0,v_2)$ (for varying $v_2$) knows the multiset of these pairs of numbers. If $v_0$ is adjacent to $v_4$, then each triangle count is too high by 1, because $(v_0,v_2)$ also counts the triangle $v_0,v_2,v_4$ and $(v_2,v_4)$ also counts the triangle $v_2,v_4,v_0$. Both these triangles don't contribute to a collection of pairs of paths $v_0, v_1, v_2, v_3, v_4$ and $v_0,v_5,v_4$ intersecting only in the endpoints and in $v_2 = v_5$. Thus the number of bad subgraphs $H(v_2 = v_5)$ is
\[ \#H(v_2 = v_5) = \sum_{v_2 \in V \setminus \{v_0,v_4\}} (p_{v_0 v_2}^1 p_{v_0 v_2}^2 - p_{v_0 v_4}^1) (p_{v_2 v_4}^1 p_{v_2 v_4}^2 - p_{v_0 v_4}^1). \]
This number is known to $(v_0, v_4)$ after 2 rounds.

Now we compute  $\#H(v_1 = v_5)$ of bad subgraphs with $v_1 = v_5$ for fixed $v_0$ and $v_4$, and varying $v_1$, $v_2$ and $v_3$.
Let $\#H(v_1=v_5,*)$ be the number of graphs obtained from a graph of type $H(v_1=v_5)$ by possibly identifying $v_2$ or $v_3$ with $v_0$. Then we have
\[ \#H(v_1 = v_5) = \#H(v_1=v_5,*) - \#H(v_1=v_5, v_0=v_2) - \#H(v_1=v_5, v_0=v_3). \]

For $\ell \in \{1,2,3\}$ every pair $(u,v)$ knows the number $p_{u v}^{\ell}$ of paths of length $\ell$ from $u$ to $v$ after $\ell-1$ rounds. Thus in particular, after 2 rounds $(v_0,v_1)$ knows $p_{v_0 v_1}^1$ and $p_{v_0 v_1}^2$, $(v_1,v_4)$ knows $p_{v_1 v_4}^3$, and $(v_0,v_4)$ knows $p_{v_0 v_4}^1$ and $p_{v_0 v_4}^2$.
Thus
\[  \#H(v_1=v_5,*) = \sum_{v_1 \in V \setminus \{v_0,v_4\}} p_{v_0 v_1}^1 p_{v_1 v_4}^3 p_{v_1 v_4}^1. \]
It is easy to see that 
\[\#H(v_1=v_5, v_0=v_2) = \binom{p_{v_0 v_4}^2}{2}, \]
because $v_0$ and $v_2$ are opposite corners of a square, and 
\[\#H(v_1=v_5, v_0=v_3) = p_{v_0 v_1}^1  p_{v_0 v_4}^1 p_{v_1 v_4}^1 (p_{v_0 v_1}^2 - 1).\]
After 2 rounds, the pair $(v_0,v_4)$ knows $\#H(v_1=v_5, v_0=v_2)$ and $\#H(v_1=v_5, v_0=v_3)$.

The computation of $\#H(v_3 = v_5)$ is completely analog.

Now we determine the number $\#H(v_2 = v_5)$.
Let $\#H(v_2 = v_5,*)$ be the number of graphs obtained from a graph of type $H(v_2=v_5)$ by possibly identifying $v_1$ with $v_3$. Then we have
\[ \#H(v_2 = v_5) = \#H(v_2=v_5,*) - \#H(v_2=v_5, v_1=v_3). \]

After 1 round $(v_0,v_2)$ and $(v_2,v_4)$ know the number of triangles in which they participate. These numbers are $p_{v_0 v_2}^1 p_{v_0 v_2}^2$ and $p_{v_2 v_4}^1 p_{v_2 v_4}^2$ respectively. After 2 rounds, $(v_0,v_2)$ (for varying $v_2$) knows the multiset of these pairs of numbers. If $v_0$ is adjacent to $v_4$, then each triangle count is too high by 1, because $(v_0,v_2)$ also counts the triangle $v_0,v_2,v_4$ and $(v_2,v_4)$ also counts the triangle $v_2,v_4,v_0$. Both these triangles don't contribute to a collection of pairs of paths $v_0, v_1, v_2, v_3, v_4$ and $v_0,v_5,v_4$ intersecting only in the endpoints and in $v_2 = v_5$. Thus the number of bad subgraphs for $v_2 = v_5$ is
\[ \#H(v_2 = v_5,*) = \sum_{v_2 \in V \setminus \{v_0,v_4\}} (p_{v_0 v_2}^1 p_{v_0 v_2}^2 - p_{v_0 v_4}^1) (p_{v_2 v_4}^1 p_{v_2 v_4}^2 - p_{v_0 v_4}^1). \]
This number is known to $(v_0, v_4)$ after 2 rounds.

Here, we hit a complication when we want to compute $\#H(v_2=v_5, v_1=v_3)$. The subgraph $H(v_2=v_5, v_1=v_3)$
is a square with a diagonal form $v_1$ to $v_2$. The other corners are $v_0$ and $v_4$. The pair $(v_0,v_4)$ does not know $\#H(v_2=v_5, v_1=v_3)$. Therefore, $(v_0,v_4)$ might not know the number of 6-cycles in which the distance from $v_0$ to $v_4$ on the cycle is 2. Luckily, we don't have to know this number, but only their sum over all $(v_0,v_4)$.
Thus instead of counting the number of $H(v_2=v_5, v_1=v_3)$ for fixed $v_0$ and $v_4$, we count this number for fixed $v_1$ and $v_2$. This number $n_{v_1 v_2}$ is 0, if $v_1$ and $v_2$ are not adjacent and $\binom{p}{2}$ for $p = p_{v_1 v_2}^2$ otherwise, and the pair $(v_1,v_2)$ knows this number. Thus instead of computing the sum of $\#H(v_2=v_5, v_1=v_3)$ over all pairs $(v_0,v_4)$, we compute the sum of $n_{v_1 v_2}$ over all $(v_1,v_2)$ obtaining the same result.
\qed \end{proof}

\section{WL[2] Does not Count 4-Cliques}

\begin{proposition}
For $k \geq 2$ every $k$-tuple (and thus also every vertex) knows the multiset of colors of all $k$-tuples of vertices of the same graph.
\end{proposition}

Note that the result does not hold for $k=1$. 

The definition of "knowing" immediately implies the following.

\begin{corollary}
For $k \geq 2$ two graphs agree in the color of one $k$-tuple, iff they agree in the multiset of  colors of all $k$-tuples.
\end{corollary}

We consider the CFI construction with the global graph being the simplest regular degree 3 graph, the 4-clique $K_4$. Assume, the vertices $v_0, \dots, v_3$ are arranged in consecutive corners of a square. For the vertices of $G$, we use double indices. The vertex $v_{ij}$ is the $j$th vertex in the $i$th corner ($i,j \in \{0,1,2,3\}$). Assume, we have assigned partition $\bf a$ to $(v_i,v_{i+1 \bmod 4})$, and partition $\bf c$ to $(v_{i+1 \bmod 4},v_i)$.
Thus partition $\bf b$ is assigned to the 4 diagonal directions. 

For every global graph $H$, WL[2] produces edge colors with the same multiplicities in the nonisomophic graphs $G$ and $\widetilde{G}$ produced by the CFI construction \cite{CaiFI92}. In fact this is the purpose of this construction.
In our case with $H = K_4$, this is immediately clear, as the the two graphs are strongly regular. Thus we have just 3 edge colors, one for original edges, one for non-adjacent disjoint pairs, and one for self-loops.

We say that two invariants are incomparable, if neither of them implies the other.

\begin{theorem}
$G$ and $\widetilde{G}$ differ in their number of occurrences of the subgraph $K_4$.
 WL[2] is incomparable with the invariant $\#4$-cliques.
\end{theorem}

\begin{proof}
 Consider $\widetilde{G}$ having the flip along the $\{v_1, v_3\}$ diagonal edge. In both, $G$ and $\widetilde{G}$ start with $v_{00}$. It is adjacent to $v_{10}$ and $v_{13}$. The vertices $v_{00}$ and $v_{10}$ are adjacent to $v_{20}$ in both $G$ and $\widetilde{G}$. Likewise, the vertices $v_{00}$ and $v_{20}$ are adjacent to $v_{30}$ in both $G$ and $\widetilde{G}$.
 Finally, $v_{30}$ is adjacent to $v_{20}$ forming a clique in $G$, but $v_{30}$ is not adjacent to $v_{20}$ forming no clique in $\widetilde{G}$. 
 Likewise, for every start in one of the vertices of $v_0$, there are two neighbors in $v_1$. Then there are twice unique common neighbors of two previously chosen vertices both in $G$ and $\widetilde{G}$. Finally in $G$ the chosen vertices in $v_3$ and $v_1$ are adjacent, but not in $\widetilde{G}$. 
 
 Considering some automorphisms, this can be verified by checking 2 cases instead of 8. The result is $G$ has 8 $K_4$, while $\widetilde{G}$ has none, even though the edge colors have the same multiplicities in $G$ and $\widetilde{G}$.
 
 That counting $K_4$ is sometimes weaker than WL[2] is trivial, e.g., take a path of length 2 and a 3-cycle.
\qed \end{proof}

\section{Difficult Cycles}
We show that WL[2] does not identify cycles of length 8. We give a clear argument about where to look for counter examples. But the actual example is created by computer. 

It is difficult to find a counter example, because WL[2] is very powerful and identifies almost all graphs. Thus we use again our example from the previous section. The non-isomorphic graphs $G$ and $\widetilde{G}$ are created with the CFI method from a tetrahedron. As the single flip in $\widetilde{G}$ can be pushed into any edge it is clear that the 2 graphs have the same number of occurrences of any subgraph that does not involve all the edges of the global graph. The global graph is a $K_4$. It has 6 edges, but the shortest closed walk through all edges has length 8. Thus the shortest cycles that might have a different count in $G$ and $\widetilde{G}$ are necessarily of length at least 8. Indeed we succeed. For all even lengths $k$ between the minimum 8 and the maximum 16 (Hamiltonian cycles), the counts in $G$ and $\widetilde{G}$ are different.

As our graphs are small, we can use a pretty brute force algorithm to count the cycles starting at a fixed vertex. For each of the two graphs, the count $C_k^v$ (the number of $k$ cycles through any given start vertex $v$)  does not depend on the chosen start vertex, because both graphs are vertex transitive. $\#k$-cycles, the total number of cycles of length $k$ is $n$ times $C_k^v$ divided by the length $k$ of the cycles.

\begin{theorem}
 WL[2] does not identify the number of 8-cycles.
\end{theorem}

\begin{proof}
 By Computer.
\qed \end{proof}

\begin{table}[t]
\label{table:cyclecount}
\begin{center}
\begin{tabular}{||r|r|r||r|r|r||}
\hline \hline
$n$ & not twisted & twisted & $n$ & not twisted & twisted \\
\hline
1 & 0 & 0 &  9 & 34368 & 33920 \\ 
2 & 48 & 48 &  10 & 91296 & 92256 \\ 
3 & 32 & 32 &  11 & 211968 & 216192 \\ 
4 & 60 & 60 &  12 & 417264 & 423216 \\ 
5 & 288 & 288 &  13 & 670464 & 674304 \\ 
6 & 1248 & 1248 &  14 & 822528 & 824448 \\ 
7 & 4032 & 4032 &  15 & 678912 & 680960 \\ 
8 & 11952 & 11688 &  16 & 284112 & 281232 \\ 
\hline
\end{tabular}
\end{center}
\caption{The number of cycles of length $n$ in $G$ and $\widetilde{G}$}
\end{table}

Table~1 is the output of the Cycle Count program. It shows that for lengths up to 7,  the number of cycles is equal in the two graphs. Starting at length 8, the number of cycles differ, i.e., WL[2], cannot count the number of 8-cycles. We knew that for this pair of graphs, the numbers have to be the same for lengths up to 6, because the graphs $G$ and $\widetilde{G}$ are constructed such that WL[2] does not detect any difference between them. It is open whether WL[2]  can always count the number of 7-cycles.

Interestingly enough, there is other evidence that the complexity changes after 7. For $k \leq 7$, Alon et al.\ \cite{AlonYZ97} can count $k$-cycles in time $O(n^{\omega})$, where $\omega < 2.373$ \cite{Williams2012} is the exponent of matrix multiplication, while
Flum and Grohe \cite{FlumG2004} show that with $k$ as a parameter the problem of counting $k$-cycles is \#W-complete.

\bibliographystyle{plain}

\begin{thebibliography}{10}

\bibitem{AlonYZ97}
Noga Alon, Raphael Yuster, and Uri Zwick.
\newblock Finding and counting given length cycles.
\newblock {\em Algorithmica}, 17(3):209--223, 1997.

\bibitem{AlzagaIP2010}
Afredo Alzaga, Rodrigo Iglesias, and Ricardo Pignol.
\newblock Spectra of symmetric powers of graphs and the {W}eisfeiler-{L}ehman
  refinements.
\newblock {\em J. Comb. Theory, Ser. B}, 100(6):671--682, 2010.

\bibitem{Babai2015}
L{\'a}szl{\'o} Babai.
\newblock Graph isomorphism in quasipolynomial time.
\newblock {\em CoRR}, abs/1512.03547, 2015.

\bibitem{Babai2016}
L{\'a}szl{\'o} Babai.
\newblock Graph isomorphism in quasipolynomial time [extended abstract].
\newblock In {\em Proceedings of the 48th Annual ACM SIGACT Symposium on Theory
  of Computing (STOC)}, pages 684--697. ACM, 2016.

\bibitem{BabaiES80}
L{\'a}szl{\'o} Babai, Paul Erd{\H{o}}s, and Stanley~M. Selkow.
\newblock Random graph isomorphism.
\newblock {\em SIAM J. Comput}, 9(3):628--635, 1980.

\bibitem{BabaiK79}
L{\'a}szl{\'o} Babai and Lud{\v e}k Ku{\v c}era.
\newblock Graph canonization in linear average time.
\newblock In {\em Proc. 20th Annual Symposium on Foundations of Computer
  Science (FOCS)}, pages 39--46. IEEE Computer Society Press, 1979.

\bibitem{BabaiM80}
L\'{a}szl\'{o} Babai and Rudi Mathon.
\newblock Talk at the South-East Conference on Combinatorics and Graph Theory,
  1980.

\bibitem{CaiFI92}
Jin-Yi Cai, Martin F{\"u}rer, and Neil Immerman.
\newblock An optimal lower bound on the number of variables for graph
  identification.
\newblock {\em Combinatorica}, 12(4):389--410, 1992.

\bibitem{Cameron80}
Peter~J. Cameron.
\newblock 6-transitive graphs.
\newblock {\em J. Comb. Theory, Ser. B}, 28(2):168--179, 1980.

\bibitem{FlumG2004}
J{\"o}rg Flum and Martin Grohe.
\newblock The parameterized complexity of counting problems.
\newblock {\em SIAM Journal on Computing}, 33(4):892--922, 2004.

\bibitem{Furer95}
Martin F{\"u}rer.
\newblock Graph isomorphism testing without numerics for graphs of bounded
  eigenvalue multiplicity.
\newblock In {\em Proc. 6th Annual ACM-SIAM Symposium on Discrete Algorithms
  (SODA)}, pages 624--631, 1995.

\bibitem{Furer2010}
Martin F{\"u}rer.
\newblock On the power of combinatorial and spectral invariants.
\newblock {\em Linear Algebra and its Applications}, 432(9):2373--2380, April
  2010.

\bibitem{GolfandK78}
Ya.~Yu. Gol'fand and M.~H. Klin.
\newblock On $k$-regular graphs.
\newblock In {\em Algorithmic Research in Combinatorics}, pages 76--85. Nauka
  Publ., Moscow, 1978.

\bibitem{KlinPR88}
M.~Ch. Klin, R.~P{\"o}schel, and K.~Rosenbaum.
\newblock {\em Angewandte Algebra}.
\newblock Vieweg \& Sohn Publ., Braunschweig, 1988.

\bibitem{Kucera87}
Lud{\v e}k Ku{\v c}era.
\newblock Canonical labeling of regular graphs in linear average time.
\newblock In {\em Proc. 28th Annual Symposium on Foundations of Computer
  Science (FOCS)}, pages 271--279. IEEE Computer Society Press, 1987.

\bibitem{Luks82}
Eugene~M. Luks.
\newblock Isomorphism of graphs of bounded valence can be tested in polynomial
  time.
\newblock {\em J. Comput. System Sci.}, 25:42--65, 1982.

\bibitem{ReadC77}
Ronald~C. Read and Derek~G. Corneil.
\newblock The graph isomorphism disease.
\newblock {\em J. Graph Theory}, 1(4):339--363, 1977.

\bibitem{WeisfeilerL76}
Boris Weisfeiler, editor.
\newblock {\em On construction and identification of graphs}, volume 558 of
  {\em Lecture Notes in Mathematics}.
\newblock Springer-Verlag, Berlin, 1976.
\newblock With contributions by A. Lehman, G. M. Adelson-Velsky, V. Arlazarov,
  I. Faragev, A. Uskov, I. Zuev, M. Rosenfeld and B. Weisfeiler.

\bibitem{Williams2012}
Virginia~Vassilevska Williams.
\newblock Multiplying matrices faster than {C}oppersmith-{W}inograd.
\newblock In {\em Proceedings of the 44th Symposium on Theory of Computing
  Conference (STOC)}, pages 887--898. ACM, 2012.

\end{thebibliography}

\end{document}